\documentclass[reqno, a4paper,12pt]{amsart}
\usepackage[dvipsnames]{xcolor}
\usepackage{tikz-cd}
\usepackage[utf8]{inputenc}
\usepackage[T1]{fontenc}
\usepackage{amsmath}
\usepackage{amsfonts,amssymb}
\usetikzlibrary{matrix}

\makeatletter
\newcommand{\vast}{\bBigg@{14.030}}
\newcommand{\Vast}{\bBigg@{18.30}}
\makeatother

\usepackage{enumerate}

\usepackage{nicematrix}
\usepackage{soul}
\usepackage[utf8]{inputenc}

\makeatletter
\ifcase \@ptsize \relax% 10pt
  \newcommand{\nano}{\@setfontsize\miniscule{3.5}{4.5}}%\newcommand{\Dcat}{\mathscr{D}}
\or% 11pt
  \newcommand{\nano}{\@setfontsize\miniscule{4.5}{5.5}}%
\or% 12pt
  \newcommand{\nano}{\@setfontsize\miniscule{4.5}{5.5}}%
\fi
\makeatother

\newcommand{\balita}{\raisebox{1.9pt}{\text{\nano$\bullet$\hspace{.7pt}}}}

\usepackage{tikz}
\usetikzlibrary{automata,arrows,topaths}
\usepackage{caption}

\usepackage[bookmarks=false]{hyperref}
\hypersetup{
         colorlinks   = true,
         citecolor   =gray,
        urlcolor = gray
}
\hypersetup{linkcolor=gray}

\newcommand{\inv}{^{-1}}

\newcommand{\Ccat}{\mathscr{C}}
\newcommand{\itemb}{\item[\balita]}

\usepackage[scr=boondoxo]{mathalfa}

\usepackage{float}
\usepackage{subcaption}
\usepackage{amsthm}
\newtheoremstyle{mytheoremstyle} % name
    {10pt}                    % Space above
    {8pt}                    % Space below
    {\itshape}                   % Body font
    {}                           % Indent amount
    {\scshape}                   % Theorem head font
    {.}                          % Punctuation after theorem head
    {.5em}                       % Space after theorem head
    {}  % Theorem head spec (can be left empty, meaning ‘normal’)

\makeatletter
\newcommand{\leqnomode}{\tagsleft@true}
\newcommand{\reqnomode}{\tagsleft@false}
\makeatother
\theoremstyle{mytheoremstyle}

\newtheorem{theorem}{Theorem}

 \newtheorem{proposition}[theorem]{Proposition}
 \newtheoremstyle{definition} % name
    {8pt}                    % Space above
    {5pt}                    % Space below
    {}                   % Body font
    {}                           % Indent amount
    {\scshape}                   % Theorem head font
    {.}                          % Punctuation after theorem head
    {.5em}                       % Space after theorem head
    {}  % Theorem head spec (can be left empty, meaning ‘norma13pt]{holl’)

 \theoremstyle{definition}

\usepackage{mathabx}

\newcommand{\Z}{\mathbb{Z}}

\newcommand{\C}{\mathbb{C}}
\newcommand{\R}{\mathbb{R}}
\newcommand{\dif}{{\mathrm{d}}}
\newcommand{\M}[1]{M_{#1}(\mathbb{C})}
\newcommand{\uni}{\mathrm{U}}

\newcommand{\ee}{\mathrm{e}}
\newcommand{\se}{_{s(e)}}
\newcommand{\te}{_{t(e)}}

\DeclareMathOperator{\End}{\mathrm{End}}

\DeclareMathOperator{\Ad}{Ad}

\DeclareMathOperator{\Tr}{Tr}
\DeclareMathOperator{\tr}{tr}

\newcommand{\hp}[1]{^{(#1)}}

\def\[#1\]{%
  \begin{align}#1\end{align}%
}

\usepackage[margin=1.1in]{geometry}
\tikzcdset{arrow style=tikz, diagrams={>={Stealth[round,length=4pt,width=4.95pt,inset=2.75pt]}}}

 \title[Comment on `Gauge networks in NCG' ]{Comment on \\`Gauge networks in noncommutative geometry' \vspace{0ex}}
 \author[C. I. Perez-Sanchez]{Carlos P\'erez-S\'anchez}
 \address{Heidelberg University, Institute for Theoretical Physics,\newline \indent
  Philosophenweg 19, 69120 Heidelberg, Germany
 }
\email{\vspace{-2ex}\href{mailto:perez.sanchez@protonmail.ch}{perez@thphys.uni-heidelberg.de, perez.sanchez@protonmail.ch}}

\makeatletter

\newcommand*\notocchapter[1]{%
  \if@openright\cleardoublepage\else\clearpage\fi
  \thispagestyle{empty}\global\@topnum\z@
  \@afterindenttrue
  \let\@secnumber\@empty
  \@makeschapterhead{#1}\@afterheading
}

\makeatother

\begin{document}

  \begin{abstract}
The article  (Gauge networks in noncommutative geometry, \textit{J. Geom. Phys.} 75 : 71--91, 2014, \cite{MvS}) that motivates this comment
provides, in particular, one answer to the following
natural question: what is noncommutative geometry on a lattice?
In the context of spectral triples,
Marcolli and van Suijlekom define in \textit{op. cit.} a Dirac operator on the lattice
and identify the corresponding Spectral Action
with the lattice Yang-Mills--Higgs system.
In this comment we show that the continuum limit of this theory is the Yang-Mills action functional,  without a Higgs scalar.
\\

\noindent
 \textsc{R\'esum\'e.}   Qu'est-ce que la g\'eom\'etrie non-commutative sur r\'eseau\,? \,\`A cette question  l'article ici comment\'e   (Gauge networks in noncommutative geometry, \textit{J. Geom. Phys.} 75 : 71--91, 2014)  apporte une des réponses possibles. Marcolli et van Suijlekom, travaillant dans le contexte des triplets spectraux, y construisent un opérateur de type Dirac pour le réseau et dérivent une théorie sur  réseau de type Yang-Mills--Higgs  à partir de l'Action Spectrale.  Ce commentaire montre que la limite  continue   de ce mod\`ele est la théorie pure de Yang-Mills  (sans aucun Higgs).
\end{abstract}

\maketitle
\vspace{-2ex}
\section{Introduction and notation}

Lattice gauge field theory \cite{montvay_muenster_1994} exists for more than half-century
\cite{Wilson} and  does not cease to inspire \cite{Cao:2023uqm}
 mathematicians and physicists.
 On the continuum, gauge theories coupled to a Higgs scalar
have, on the other hand, an interesting geometrical interpretation in terms of certain spectral triples (see Sec. \ref{sec:STr} for a definition)
referred to as almost-commutative geometries   \cite{Chamseddine:2006ep, Barrett:2006qq,vanSuijlekom:2024jvw}.
It is natural to ask whether the Spectral Action \cite{Chamseddine:1996zu} ---
a functional that describes the dynamics of
noncommutative geometries (thus the physics of almost-commutative geometries) --- can yield a lattice gauge-Higgs theory.\\

Since we will work below with a discrete setting (for which we shall provide formal definitions), we
only need the heuristics of almost-commutative geometries $M \times F$.
They have a manifold factor $M$ (which leads to an infinite-dimensional
commutative algebra) and a noncommutative factor $F$
modelled on a finite-dimensional
matrix algebra, which being  much smaller than $C(M)$, justifies calling  $M\times F$
`almost-commutative'.
In \cite{MvS} almost-commutative geometries are formulated
in a discrete setting by using,
instead of an (even-dimensional) spin manifold $M$,
an embedded directed graph $\Gamma \subset M$. In \textit{op. cit.}
a gauge-Higgs functional on $\Gamma$ emerges from the Spectral Action in
the way described in Section \ref{sec:SA_siec}.
This comment proves that the lattice gauge-Higgs functional of
\cite[Thm. 28]{MvS}, see Eq. \eqref{trzydziesci} below, has the continuum limit of a pure gauge theory.
\\

\textsc{Notation.} To ease the comparison with the main source,
the equation \textit{numbers} in this comment match those of \cite{MvS}.
Whenever new equations need labels  here, we use  letters instead.

% \vspace{-2ex}

\section{Spectral triples on graphs} \label{sec:STr}

According to \cite[Def. 2]{MvS},
a \textit{finite spectral triple} $(A,H,D)$ consists of
the following data: a finite-dimensional (unital) C${}^*$-algebra $A$,
an inner product space $H$ together with a $*$-action $\lambda$
of $A$ on $H$, i.e. such that $\lambda(a^*) $
is the adjoint  $\lambda(a)^*$  of $\lambda(a)$ for each $a\in A$; and finally, $D:H\to H$, which is a symmetric linear operator (the Dirac operator).
\\

\textit{In at least one manner}, finite spectral triples can form a category.
We denote by $\mathscr C$ the category of finite spectral triples
with the morphism structure given in \cite[Def. 2]{MvS}, which we now reproduce.
Given  two spectral triples, $X=(A,H,D)$ and  $X'=(A',H',D')$,
the set $\hom_\Ccat (X,X')$ of morphisms
from $X$ to $X'$ consists  of pairs $(\phi,L)$ of unital $*$-algebra maps $\phi: A \to A'$
along with unitary maps $L:H\to H'$. As part of this definition,
$\phi$, $D$ and $L$ obey two compatibility conditions: first, for all $a\in A$,
\[ \phi(a) = L a L^*,\]

(or
more precisely  $\lambda' \circ \phi = \Ad_L\circ  \lambda $, where $\lambda$
and $\lambda'$ are, respectively,
  the actions of $A$ on $H$ and of $A'$ on $H'$%,  see footnote \ref{footn:Def2} here
) and, secondly,
\[ D'= L D L^*, \]
 which is crucial for the argument of the present comment.
\\

% {   }

Let $\Gamma\hp0$ and $\Gamma\hp 1$ denote vertices and edges, respectively, of a directed graph $\Gamma$ that is embedded in $M$. For fixed $N\in\Z_{\geq 2}$, instead of the ordinary algebra $C(M)\otimes M_N(\C)$
of the archetypical almost-commutative geometry that yields $\uni(N)$-Yang-Mills theory,
the algebra of the discrete geometry is, according to \cite[Sec. 4]{MvS},
\[ \notag M_N(\C)^{\Gamma\hp 0}=\bigoplus_{v\in \Gamma\hp0} M_N(\C) .\]

This algebra naturally acts  on the second factor of
the Hilbert space
$\mathscr S \otimes (\C^N)^{\Gamma\hp 0 } $,
being $\mathscr S$  the fiber of the spinor bundle of $M$.
(The space
$\mathscr S \otimes (\C^N)^{\Gamma\hp 0 } $ is the
discrete version of the $L^2$-space of spinors on $M$.)
In order to construct a Dirac operator that fills the third slot in
\[ \notag ( M_N(\C)^{\Gamma\hp 0}, \mathscr S \otimes (\C^N)^{\Gamma\hp 0 }, \,\, \balita\,\,)\]

and that turns the pair into a spectral triple%
% \footnote{According to \cite[Def. 2]{MvS},
% a finite spectral triple $(A,H,D)$ consists of
% the following data: a finite-dimensiona (unital) C${}^*$-algebra $A$,
% an inner product space $H$ together with an action $\lambda$
% of $A$ on $H$ that is involutive, i.e. that $\lambda(a^*) $
% is the adjoint  $\lambda(a)^*$  of $\lambda(a)$ for each $a\in A$; finally $D:H\to H$ is a symmetric linear operator (the Dirac operator).\label{footn:Def2}}
, we recall some notation
 and results of \cite[Sec. 2]{MvS} next.

From now on $M=\R^4$ and $\Gamma=(\Gamma\hp0,\Gamma\hp 1) $  will
be the four-dimensional square lattice. Let us denote its lattice spacing by $l>0$, so vertices are $\Gamma\hp 0=(l\mathbb  Z)^4 $ and edges
$\Gamma\hp 1$
consist of nearest neighbours pointing in `positive direction',
i.e. any $e\in \Gamma\hp 1$ is so directed, that\footnote{Since a sporadic
reference to quiver representations is made below, we remark that our notation
is
$e=(s(e),t(e))$,
so $s$ means `source' of the arrow or oriented edge $e$, and $t$ means its `target',
not its `tail', as elsewhere in the  quiver representations literature, e.g. in
\cite{bookQuivRep}.}
\[t(e) -s(e) \in \{ (l,0,0,0), (0,l,0,0),(0,0,l,0), (0,0,0,l)\}.  \tag{A}\label{edgeslist}
\]

% {   }

To define the Dirac operator (henceforth $\mathscr S\cong \C^4$)
\[D_{\Gamma,L}: \mathscr S \otimes (\C^N)^{\Gamma\hp 0 }\to  \mathscr S \otimes (\C^N)^{\Gamma\hp 0 }  \notag\]

the authors \textit{choose} a representation of the lattice $\Gamma$ in the category $\Ccat$, which remains fixed.
This is to be interpreted in the spirit of quiver $\Ccat$-representations, namely as a functor from the path category of $\Gamma$ to $\Ccat$.
Spelled out, a $\Ccat$-representation of $\Gamma$ is an association $X:\Gamma\hp0 \to \Ccat$ of a finite spectral triple $X_v=(A_v,H_v,D_v)$ to each vertex $v$  of
$\Gamma$, as well as of a $\Ccat $-morphism $(\phi_e,L_e)$ to each edge $e=(s(e),t(e))= s(e) \stackrel{e}{\longrightarrow} t(e) $ of $\Gamma$, organised as  \[ \big \{ (\phi_e,L_e)  \in \hom_\Ccat ( X\se,X\te) \big\}_{e\in \Gamma\hp1} .  \notag \]

% {   }

Seen as an $\End (\mathscr S\otimes \C^N)  $-valued square matrix $D_{\Gamma,L} \in M_{|\Gamma\hp 0| } \big( \End (\mathscr S\otimes \C^N)  \big) $ whose entries are indexed by the vertices of $\Gamma$, the Dirac operator $  D_{\Gamma,L} $
has an off-diagonal piece $D^{\text{\tiny \textsc{ym}}}(L)$
and a diagonal piece.
More precisely, if $\gamma$ is the  chirality or $\mathbb Z_2$-grading on $\mathscr S$ (usually denoted $\gamma_5$),
\[ \notag (D_{\Gamma,L})_{v,w} = [ D^{\text{\tiny \textsc{ym}}}(L)  ]_{v,w} + \delta_{v,w} \gamma \otimes D_v , \qquad  (v,w\in \Gamma\hp 0) \]

where $\delta_{v,w}=0$ if $v\neq w$ and $\delta_{v,v}=1$ for all $v$. The piece $ D^{\text{\tiny \textsc{ym}}}(L)$
is defined in terms of parallel transports (presupposing a spin connection that in this comment does not play any role), gamma matrices along the edges,
and depends also the data $L_e$ associated
to each edge $e$ by the chosen $\Ccat$-representation of $\Gamma$. The exact form of $ D^{\text{\tiny \textsc{ym}}}(L)$ is here irrelevant;
we only stress that it is thanks to this piece that  the  Spectral Action (see below)
yields Wilson's $\uni(N)$-Yang-Mills lattice action\footnote{The exact
form for $ W(L)$ is given as a sum over
plaquettes $p$, i.e. length-4 non-backtracking loops $p=e_4e_3\bar e_2\bar e_1$ (with the bar denoting the reverse edge), as
 \[ W(L)  =-  \sum_{p} \tr \big( L_{ e_4}^* L_{ e_3}^*L_{e_2} L_{e_1}  + L_{ e_1}^* L_{ e_2}^*L_{e_3} L_{e_4}  \big) \notag .\]
}, $ W(L) $, while crossed products of the diagonal piece
($\gamma\otimes D$) with $D^{\text{\tiny \textsc{ym}}}(L)$ yield
the lattice gauge-Higgs coupling terms.

\section{The Spectral Action on the lattice} \label{sec:SA_siec}

The Spectral Action used in \cite{MvS}  depends on the lattice constant $l$ and on
the previous Dirac operator $D_{\Gamma,L}$ and reads\footnote{Readers familiar the
Spectral Action  in the form $\Tr f(D/\Lambda)$
will identify the energy scale $\Lambda=1/l$ and observe that the function
$f$ here is a monomial, $f(x)=x^4$, in contrast to the bump function used by
Connes-Chamseddine \cite{Chamseddine:1996zu}.}
\[S(D_{\Gamma,L})= l^4 \Tr [  (D_{\Gamma,L})^4 ], \notag  \]

where $\Tr$ is the trace of $M_{|\Gamma\hp 0| } \big[ \End (\mathscr S\otimes \C^N)  \big] $.
We do not need further details, since according to \cite[Thm. 28]{MvS},\
up to an additive constant:
\[ \label{trzydziesci}  \tag{13}
S(D_{\Gamma,L})& = W(L)  \\[1ex]  &  +
4 l^4 \sum_{v \in \Gamma\hp0 } \tr D_v^4  + 16l^2 \sum_{e \in \Gamma\hp 1}  \tr \big( D^2\se + D^2\te    -
 L_e^* D\te L_e D\se \big)\,,\notag
\]

\noindent
where $\tr$ is the trace
of $\M N$. It is written in lowercase
to distinguish it from the trace  that appears in the Spectral Action above,
but it is unnormalised ($\tr 1=N$).
% Although it is not relevant in the sequel,
% the reader might want to understand the origin expression.
% Eq.  \eqref{trzydziesci}
% all types
% of length-4 paths that emerge when one traces $ (D_{\Gamma,L})^4 = [D^{\text{\tiny \textsc{ym}}}(L)  + \gamma \otimes D]^4$.
% From $[D^{\text{\tiny \textsc{ym}}}(L)]^4$ one gets
% non-backtracking paths of length four, i.e. plaquettes,
% and backtracking paths which contribute $\tr (L_{e_1}L_{e_1}^*L_{e_2}L^*_{e_2}) =1$,
% i.e. the hidden
% constant.  If $\gamma \otimes D$ appears twice,
% one gets the sum over edges in the second line of
% Eq.  \eqref{trzydziesci}, the numerical factors
% coming from tracing certain $4\times 4$ gamma-matrices.
% And the diagonal term $\gamma \otimes D$ does not occur,

Let $\epsilon_\mu$ be the $\mu$-th standard basis vector of $\mathbb Z^4$.
The unitary matrix $L_{e} $ corresponding to
an edge  $e=(v,v+l \epsilon_\mu)$  can be expressed as $ L_e=\ee^{i l A_\mu(v) }$ ($l\to0$)
for some Hermitian matrix $A_\mu(v) \in \M N$, which in the continuum
  becomes an ordinary Yang-Mills connection. If $F_{\mu\nu}$ is its curvature,
according to \cite[Prop. 29]{MvS},
\[ \notag
W(L) \to
\int_{\R^4}  \tr \big(  F_{\mu\nu} F^{\mu\nu} )  \qquad ({l\to 0}),
\]

(here and below, repeated Greek indices are implicit sums \`a la Einstein).

% We now address the the second line of Eq. \eqref{trzydziesci}.
Similarly, \cite[Prop. 29]{MvS}  claims that the collection of operators  $\{D_v\}_{v\in\Gamma\hp0}$  gives rise to a Hermitian Higgs field $\Phi(x)$,  $x\in \R^4$, in terms of which
the Higgs potential and Higgs-gauge coupling are derived. More precisely, that
the second  line of Eq. \eqref{trzydziesci} yields
$ \int _{\R^4}   ( \mathcal L^{\textrm{H}} + \mathcal L^{\textrm{gH}}) \dif ^4x $
 in the limit $ l\to 0$,
where (below square brackets are commutator):
% \[ \notag
% \text{Second line of Eq. }
% \eqref{trzydziesci}
% \to
% \int _M   ( \mathcal L^{\textrm{H}} + \mathcal L^{\textrm{gH}})  \qquad (l\to 0) \]
% where
\[ \notag
\mathcal L^{\textrm{gH}} & = 8 \tr \big\{  \big( \partial_\mu \Phi -[i A_\mu, \Phi] \big) \big( \partial^\mu \Phi -[i A^\mu, \Phi] \big)  \big\} \,,   \\
\mathcal L^{\textrm{H}} & =
32 l^{-2 } \tr \Phi^2
+ 4  \tr \Phi^4   \,. \notag
\]

These two Lagrangians were identified in   \cite[Prop. 29]{MvS} --- and based on it, also in  \cite[Sec. 9.4.4]{Marcolli:2018uea} --- as the
gauge-Higgs and Higgs potentials, respectively.

\section{Our main result}

Now we proceed with the main aim of this comment.

\begin{proposition} Up to an additive constant,
 the continuum limit of Eq. \eqref{trzydziesci},
 and therefore that of the
Spectral Action of \cite[Sec. 4]{MvS},  is pure Yang-Mills, $\int_{\R^4} \tr (F_{\mu\nu}F^{\mu\nu})$.
\end{proposition}

\begin{proof} Instead of looking at the Lagrangians
$
\mathcal L^{\textrm{H}} $ and $ \mathcal L^{\textrm{gH}}$ given above,
we examine in detail their origin, namely Eq. \eqref{trzydziesci}. Since $W(L)$ yields already the pure
Yang-Mills action in the limit $l\to0$, it remains to prove that
the terms  $ \tr D_v^4 $, $  \tr D\se ^2$, $ \tr D\te^2$  and $ \tr ( L_e^* D\te L_e D\se ) $ that are
summed in the second line of Eq. \eqref{trzydziesci},
actually do not depend on the vertices $v$ nor on the edges $e$.
This is so because

 \begin{itemize}
  \itemb the family $\{D_v: H_v \to H_v \}_{v\in \Gamma \hp 0}$  of Dirac operators and \vspace{.4ex}
  \itemb the family   $\{L_e: H_{s(e)} \to H_{t(e)} \}_{e\in \Gamma \hp 1}$ of unitary matrices,
 \end{itemize}
 are not independent, but come from a $\Ccat$-representation of $\Gamma$ (see discussion just below
 Eq. \eqref{edgeslist} above)
 and are therefore  interwoven  by \cite[Eq. 2]{MvS}. Indeed, applying
 this equation to each edge, one observes that
%  to each edge $e\in \Gamma\hp1$, one obtains
 \[  D_{t(e)}=  L_e D_{s(e)} L_{e}^*  \qquad
  \text{for all } e \in \Gamma\hp 1,
 \notag \]
 \noindent
which, together with the cyclicity of the trace, means that
 \[    \tr \big(   D^2\te   -
 L_e^* D\te L_e D\se  \big) =
\tr \big[   D^2\te    -
   D\te  (L_e D\se L_e^*)  \big] =  0\notag
 \]
 \noindent
for all $e$. Due to this relation,
in the second line of Eq. \eqref{trzydziesci},  some terms cancel out as follows:
 \[ \notag &\,\,
 4 l^4 \sum_{v \in \Gamma\hp0 } \tr D_v^4  + 16l^2 \sum_{e \in \Gamma\hp 1}  \tr \big( D^2\se + D^2\te    -
 L_e^* D\te L_e D\se \big)  \\  =&\,\,
4 l^4 \sum_{v \in \Gamma\hp0 } \tr D_v^4  + 16l^2 \sum_{e \in \Gamma\hp 1}  \tr \big( D^2\se \big)  \label{B} \tag{B}  \\
=&\,\,
 \sum_{v \in \Gamma\hp0 }\tr  \big(  4 l^4  D_v^4  + 64  l^2    D_v^2 \big) ,
\notag \]

\noindent
where, in order to obtain the last equation, the fact that there are four
outgoing edges per vertex is used   (i.e. that for each $v\in\Gamma \hp0$,
the preimages of $s$ satisfy $ |s\inv (v)|=4$).
\\

We now prove the vertex-independence of the summands in
Eq. \eqref{B}. Pick any two vertices $x,y \in  \Gamma \hp 0=(l\mathbb  Z)^4$ of the rectangular lattice
and let
$ x \wedge y \in (l \mathbb Z)^4$ be given by
\[x \wedge y: =\big(\hspace{-.3pt}\min\{x_1,y_1\},\min\{x_2,y_2\},\min\{x_3,y_3\},\min\{x_4,y_4\} \hspace{1.3pt}\big). \notag  \]

By construction, there exists at least two paths on $\Gamma$,
$\alpha_1$ and $\alpha_2$,
that
consist of edges of $\Gamma\hp 1$, i.e. obeying \eqref{edgeslist}, with start and end points  [$s(\alpha_i)$ and $t(\alpha_i)$, respectively] given by
 \[&&s(\alpha_1)&=x \wedge y,  && & s(\alpha_2) &=x \wedge y, && \nonumber \\
 &&t(\alpha_1)&=x,   && \hspace{-13ex}&
 t(\alpha_2)&=y. && \notag
\]

 For each $i=1,2$,
 let
 $L_{\alpha_i}$ denote the matrix associated to $\alpha_i$ by the
  representation of $\Gamma$.
  We recall that a $\Ccat$-representation of $\Gamma$ is a functor from the free category of $\Gamma$ to $\Ccat$. Hence, both
 $L_{\alpha_1}  $ and $L_{\alpha_2}$ are determined by functoriality  in terms of the ordered product of the $L$-matrices on the edges
 of $\Gamma$ that form each path
 $\alpha_i$ (e.g. if  $\beta= e_n\cdots  e_2 e_1$ is a path, i.e. an edge composition
 that is legal $s(e_i)=t(e_{i-1}), i=2,\ldots,n$, then  $L_\beta = L_{e_n} \cdots L_{e_2} L_{e_1}$).  Therefore,
 $L_{\alpha_1}  $ and $L_{\alpha_2}$ are unitary, since
 so is $L_e$ for each $e\in\Gamma\hp1$. Finally, applying \cite[Eq. 2]{MvS}
 for all edges composing  $\alpha_1$ or $\alpha_2$ implies
 \[ D_{x} =L_{\alpha_1}  D_{x \wedge y} L_{\alpha_1}^* \quad  \text{ and } \quad D_{y} =L_{\alpha_2}  D_{x \wedge y} L_{\alpha_2}^* . \notag \]

Therefore, using the unitariness of each $L_{\alpha_i}$,
\[ \nonumber
\tr \big[ (D_{x})^m \big] & = \tr \big [
(L_{\alpha_1}  D_{x \wedge y} L_{\alpha_1}^*)^m
\big]  \\ \notag  &
= \nonumber
 \tr \big [
L_ {\alpha_1}  (D_{x \wedge y} )^m L_{\alpha_1}^*
\big]
  \\ &
= \tr \big [  ( D_{x \wedge y})^m \big] \qquad \qquad \quad (m \in \Z_{>0})
\nonumber
% \tag{C} \label{star}
\\ & \nonumber
=
\tr \big [
L_{\alpha_2} (D_{x \wedge y} )^m L_{\alpha_2}^*
\big] \nonumber
  \\ &= \tr \big [
(L_{\alpha_2}  D_{x \wedge y} L_{\alpha_2}^*)^m
\big]
= \tr \big[ (D_{y})^m \big]\,. \nonumber
\]

\noindent
Hence, powers of the Dirac operator have  a constant trace, in the sense of the map
\[ \Gamma \hp 0 \ni v\mapsto \tr \big( D_{v}^m \big) \qquad\text { (for fixed $m$)} \notag  \]

\noindent
being a constant.  Therefore $ \tr  ( D_{v}^m   ) $  will escape from the sum $\sum_{v\in \Gamma\hp0}  \tr  ( D_{v}^m  ) = \tr  ( D_{v}^m  )  \times | \Gamma\hp0 | $ (see footnote
\ref{footn:sums} below) appearing in the  Spectral Action,
and therefore from the integral over $\R^4$ in the continuum limit.
\\

We conclude from this analysis, that
\textit{all} summands in Eq. \eqref{trzydziesci} in which $\{D_v\}_{v\in \Gamma\hp 0}$ occur [those in Eq. \eqref{B}] are non-dynamic, as they are
an additive  constant (i.e. independent from the point $x\in \R^4$).
Hence, those sums over vertices and edges contribute, in the best  case\footnote{\label{footn:sums}A sum like
$ \sum_{v\in  \Gamma\hp0} \tr D_v^m $ might converge, since
$\tr D_v^m$ is independent from $v$, but the value of $\tr D_v^m$ can be set (for all $v$)
to some function of $l$ and the number of vertices, in such a way that
convergence is obtained (for example, if such sum is taken in the sense
of the $n \to \infty$ limit of \[  \sum_{v\in \Gamma\hp0 _{< n}} \tr D_v^m
\text{ with }  \Gamma\hp0 _{< n} := \big\{ (x_1,x_2,x_3,x_3) \in (l\Z)^4 :  |l x_i| < n   \text{ for each } i \big\} \notag\]
and if $D_v$ is such that $\tr D_v^m $ is proportional to $\Big|\,
 \Gamma\hp0 _{< n} \, \Big|\inv  $,  which is not a function of $v$).},
a finite additive constant to the Spectral Action $l^4 \Tr [  (D_{\Gamma,L})^4 ]$, which, up to constants, reads then
$\int_{\mathbb R^4} F_{\sigma\rho} F^{\sigma\rho} $ ($l\to0$).
\end{proof}

\subsection{Closing remarks}

The constancy of the \textit{pure} Higgs potential was already sketched (with
other argument) in \cite{NCGquivers}. This motivated a lattice Yang-Mills--Higgs model
with a dynamic Higgs after:  (i) changing  Dirac operator, (ii) using an actual quiver
instead of $(l\Z)^4$, and (iii) replacing the target category. Such setting
is also derived from the Spectral Action.
But of course, another, more conservative model motivated by
\cite{MvS} is still worth exploring ---
all the more so, due to the remarkable feature
of the potential for $\Phi$, which contains a
 quadratic-quartic potential, despite the
 Spectral Action being purely quartic.

 {
\footnotesize
 \section*{Acknowledgements}
 This work was mainly supported by the Deutsche
Forschungsgemeinschaft (DFG, German Research Foundation) under
Germany’s Excellence Strategy EXC-2181/1-390900948 (the Heidelberg
\textsc{Structures} Cluster of Excellence) and at the beginning by the European Research
Council (ERC) under the European Union’s Horizon 2020 research and
innovation program (grant agreement No818066).

%  \bibliographystyle{plain}
%  \bibliographystyle{ieeetr}
% \bibliography{biblio_komentarz}

\end{document}